\documentclass{article}
\usepackage[utf8]{inputenc}
\usepackage[totalwidth=400pt, totalheight=600pt]{geometry}
\usepackage{booktabs} 
\usepackage{amsmath,amssymb,amsthm,tikz,comment}
\usepackage{hyperref}
\usetikzlibrary{patterns,shapes,positioning}
\usepackage{mathtools}
\usepackage[ruled]{algorithm2e} 

\SetAlFnt{\small}
\SetAlCapFnt{\small}
\SetAlCapNameFnt{\small}
\SetAlCapHSkip{0pt}
\IncMargin{-\parindent}

\newcommand{\XP}{\mathbf{XP}}

\newcommand{\FPT}{\mathbf{FPT}}

\newcommand{\Poly}{\mathbf{P}}

\newcommand{\NP}{\mathbf{NP}}

\newcommand{\W}{\mathbf{W}}

\newcommand{\cE}{\mathcal{E}}

\newcommand{\cR}{\mathcal{R}}

\newcommand{\cO}{\mathcal{O}}

\newcommand{\cA}{\mathcal{A}}

\newcommand{\yields}[1]{\overset{#1}{\rightarrow}}

\newtheorem{theorem}{Theorem}
\newtheorem{lemma}[theorem]{Lemma}
\newtheorem{claim}[theorem]{Claim}

\newtheorem{proposition}[theorem]{Proposition}
\newtheorem{observation}[theorem]{Observation}

\title{On the parameterized complexity of manipulating\\ Top Trading Cycles}  
\author{
William Phan \\ North Carolina State University \\\href{mailto:wphan@ncsu.edu}{wphan@ncsu.edu}
\and
Christopher Purcell \\ Aalto University \\ \href{mailto:christopher.purcell@aalto.fi}{christopher.purcell@aalto.fi}
}

\begin{document}
\maketitle

\begin{abstract}

We study the problem of exchange when 1) agents are endowed with heterogeneous indivisible objects, and 2) there is no money. In general, no rule satisfies the three central properties Pareto-efficiency, individual rationality, and strategy-proofness \cite{Sonmez1999}.
Recently, it was shown that Top Trading Cycles is $\NP$-hard to manipulate \cite{FujitaEA2015}, a relaxation of strategy-proofness.
However, parameterized complexity is a more appropriate framework
for this and other economic settings. Certain aspects of the problem - number of objects each agent brings to the table, goods up for auction, candidates in an election \cite{consandlang2007}, legislative figures to influence \cite{christian2007complexity} - may face natural bounds or are fixed as the problem grows. We take a parameterized complexity approach
to indivisible goods exchange for the first time. Our results
represent good and bad news for TTC. When the size of 
the endowments $k$ is a fixed
constant, we show that the computational task of manipulating TTC can be performed
in polynomial time. On the other hand, we show that this 
parameterized problem is $\W[1]$-hard, and therefore unlikely
to be \emph{fixed parameter tractable}.

\end{abstract}

\section{Introduction}

In many economic environments, agents are endowed with heterogeneous 
indivisible objects, exchange is desirable, and there is no money.
For example, workers trading
shifts/tasks/assignments, users sharing time blocks on a supercomputer,
etc.
A rule or mechanism
recommends for each possible profile of preferences and endowments
a re-allocation of the objects. The general program is to define desirable properties (axioms)
and design rules that satisfy as many of them as possible. Three central
and well-studied properties are \emph{Pareto-efficiency} (no rearrangement
could make all agents at least as well off, and some better off),
\emph{individual rationality} (no agent is worse off than they
started), and \emph{strategy-proofness} (no agent is better off reporting
a lie than their true preference). Unfortunately, in this environment
and many others, there is no rule satisfying all three \cite{Sonmez1999}. This motivates the study of properties which
are relaxations of strategy-proofness.

We focus on the Top Trading Cycles (TTC) mechanism due to Gale.
TTC is strategy-proof when the endowments are of size 1,
but when the endowments are multiple, manipulation is possible.
It was recently shown that TTC is $\NP$-{hard to manipulate}.
This result suggests that there may not be an incentive to
manipulate TTC, as agents have bounded computational resources.
However, the result could be very misleading to policy makers.
As we will see, manipulating TTC can be done in time approximately
$n^k$, where $n$ is the number of goods and $k$ is the size of
the endowments.
This does not contradict \cite{FujitaEA2015} because the
size of the endowment is part of the input to the problem 
in that paper. Indeed, the hardness reduction makes
the implicit assumption that an agent may have a number of
goods that grows with the number of agents.

Still, when $k$ is large, the algorithm is not practical.
If TTC could be manipulated in time $f(k)n^c$ for some 
function $f$ and constant $c$ (i.e. if the problem
were {\em fixed parameter tractable}) Then TTC would
cease to be an attractive rule in the general case.
The parameterized complexity of manipulating TTC is therefore
an important and interesting question.

Our previously mentioned algorithm is bad news for TTC, but
our main result represents good news. We show that
manipulating TTC is $\W[1]$-hard (which we define properly later)
and therefore very unlikely to be fixed parameter tractable.


\subsection*{Related Literature}

We discuss two bodies of related work: the progression of the study
of indivisible objects exchange in the economics literature, and recent
work in computational social choice. 

In 1974, Shapley \& Scarf introduced the problem of exchanging indivisible
objects without money, also known as the Housing Market \cite{Shapley1974}.
Each agent is endowed one object, may consume one object, and has
strict preferences over all objects. They showed that the Top Trading
Cycles algorithm (attributed to David Gale) could be used to
compute a core allocation. It turns out that the core is unique, and
the rule derived from recommending the core for each preference profile
is the only \emph{Pareto-efficient}, \emph{individually rational},
and \emph{strategy-proof} rule \cite{Anno2015,Roth1977weak,ma1994strategy,miyagawa2002strategy,sethuraman2016alternative,Sonmez1999,svensson1999strategy}. 

Subsequently, the literature considered various generalizations of
the model and/or applications of TTC: the case of no ownership \cite{hylland1979efficient,zhou1990conjecture,svensson1999strategy},
the case where some agents may own nothing (generalizing the two previous
cases) \cite{abdulkadirouglu1999house,sonmez2010house,svensson1999strategy,papai2000strategyproof,pycia2017incentive},
fairer probabilistic rules \cite{abdulkadirouglu1998random,Aziz2015generalizing,athanassoglou2011house,BasteckFair,bogomolnaia2012probabilistic,bogomolnaia2001new,carroll2014general,budish2013designing,che2010asymptotic,harless2017endowment,Liu2016ordinal}
allowing for indifferences in preferences \cite{alcalde2011exchange,aziz2012housing,bogomolnaia2005strategy,jaramillo2012difference,plaxton2013simple,quint2004houseswapping,saban2013house,sonoda2014two},
School Choice \cite{abdulkadirouglu2003school,dur2012characterization,morrill2015two,dur2017competitive,morrill2015making,morrill2013alternative},
and dynamic environments \cite{kurino2014house,schummer2013assignment}.
Several authors considered manipulation not by preference misreport
but by the merging/splitting/withholding of endowments \cite{atlamaz2007manipulation,bu2014merging}. 

Our paper considers the case when each agent may be endowed with multiple
objects \cite{papai2003strategyproof,papai2007exchange,sonoda2014two,todo2014strategyproof}.
As mentioned, there is no rule satisfying all three properties \cite{Sonmez1999}.
In response to this, \cite{papai2007exchange} weakens the \emph{Pareto-efficiency}
requirement to \emph{range-efficiency }and characterizes the resulting
family of rules on a large preference domain. In a complementary manner, \cite{FujitaEA2015b} shows that in the Lexicographic domain of
preferences ATTC satisfies the properties when \emph{strategy-proofness}
is weakened to \emph{NP-hard to Manipulate}. An immediate corollary
of their result is the extension of the statement to larger domains.
Other authors consider an environment where objects have types \cite{klaus2008coordinate,konishi2001shapley,monte2015centralized,mackin2015allocating,sikdar2017mechanism},
or where there is no ownership \cite{budish2011combinatorial,budish2012multi}. 

The idea of studying the complexity of manipulation was proposed by
\cite{bartholdi1989computational} in response to \cite{gibbard1973manipulation,satterthwaite1975strategy}\textemdash the
latter showing that, in the environment of voting, requiring \emph{strategy-proofness}
leads to dictatorship. We refer the reader to surveys in the subsequent
computational social choice literature \cite{Conitzer2016Barriers,faliszewski2010using,faliszewski2010ai},
and highlight works that take the parameterized complexity approach
\cite{betzler2008fixed,betzler2010parameterized}.

\section{Preliminaries}

Let $N$ be a set of {\em agents} and let $\mathcal{O}$ be a
set of {\em objects}. Let $\omega=\{\omega_i\}_{i\in N}$ be a
set of subsets of $\mathcal{O}$ such that $\omega_i \cap
\omega_j = \emptyset$ for all $i \not= j$ and $\cup \omega_i
= \mathcal{O}$; we call $\omega_i$ the {\em endowment} of
agent $i$.  If a good $\alpha$ is in $\omega_i$ then we say
that agent $i$ is the {\em owner} of $\alpha$ and that 
$a(\alpha)=i$. Let $\mathcal{R}$ be the set of all relations
over $2^{\mathcal{O}}$ that are complete, transitive and
anti-symmetric.  Let $R=\{R_i\}_{i \in N}$ be an element of
$\mathcal{R}^{|N|}$; we call $R_i$ the {\em preference
relation} of agent $i$, and $R$ the {\em preference profile}. We denote the strict component of
$R_i$ by $P_i$, i.e. $X P_i Y$ if and only if $X R_i Y$ and
$\neg Y R_i X$. 
We say that $\cE=(N,\mathcal{O}, \omega,
R)$ is an {\em economy}. If $|\omega_i|=1$ for all $i$ we say $E$
is a {\em housing market} and otherwise a {\em generalised
housing market}.  If $z=\{z_i\}_{i \in N}$ is a set of
disjoint
subsets of $\mathcal{O}$ such $\cup z_i = \mathcal{O}$
we say that z is an {\em allocation} for the economy $\cE$.
Note that the endowment is an allocation. A {\em rule}
$\phi: \mathcal{R}^{|N|} \rightarrow Z$ recommends an
allocation given a particular preference profile.
We denote by $\phi_i(R)$ the allocation of agent $i$ under
$\phi$ at $R$; if $\phi(R)=z$ then $\phi_i(R) = z_i$.

\paragraph*{Properties of rules}

Following standard notation
we write $(R'_i,R_{-i})$ to be the preference
profile obtained from $R$ by replacing $R_i$ with $R'_i$. We say
that $R'_i$ is a {\em misreport} for agent $i$.
A misreport $R'_i$ is {\em beneficial} under $\phi$ if $\phi_i(R'_i,R_{-i}) P_i \phi_i(R)$.
A rule $\phi$ is {\em strategy-proof} if for all economies,
no agent has a beneficial misreport under $\phi$. 
We emphasise that this property implies that no agent
can lie even if they have full information about the
preferences of the other agents.  One trivial example of a
strategy-proof rule is the ``no deal'' rule $\phi(R) =
\omega$, but this rule is clearly sub optimal. We say that
an allocation $z$ is {\em Pareto-optimal} for $R$ if 
for any $z'$ we have that
$z_i R_i z'_i$ and for at least one agent $j$ we do not have
that $z'_i R_i z_i$. We say that a rule is {\em Pareto-efficient} (PE) if
it always recommends a Pareto-optimal allocation. 
If an agent might be worse off after the trade
according to their own preference relation, there is no
incentive to take part. A rule is said to be {\em
individually rational} (IR) if $\phi_i(R) R_i \omega_i$ for each agent $i$.

\paragraph*{Graph theory}

In order to describe the rule that is the focus on 
this paper, and our results, we require some definitions from graph theory.
We follow the definitions in \cite{BondyMurty2008}, but
we now recall some important notions.
A (directed) {\em walk} in a graph is an ordered multiset $(v_1,e_1,v_2,e_2,\ldots,e_k,v_{k+1})$ where $v_i$ is a vertex and $e_i$ is a (directed) edge from $v_i$ to $v_{i+1}$ for $1\leq i < k$.
A path is a walk where no vertex is repeated. A cycle is a path plus an edge from $v_k$ to $v_1$.
A {\em clique} in a graph $G$ is a set of vertices $C$
such that there is an edge between every pair of vertices in $C$.
A {\em proper colouring} (or simply a colouring) of a
graph $G$ is an assignment of colours to its vertices 
such that no edge joins two vertices of the same colour.

\begin{figure}[b]

\begin{center}
\begin{tikzpicture}

\tikzstyle{every node} = [circle, fill=white]

\node (a) at (0,2) {$\alpha$};

\node (b) at (2,2) {$\beta$};

\node (g) at (2,0) {$\gamma$};

\node (d) at (0,0) {$\delta$};

\node[text=gray] (a2) at (4,2) {$\alpha$};

\node (b2) at (6,2) {$\beta$};

\node[text=gray] (g2) at (6,0) {$\gamma$};

\node (d2) at (4,0) {$\delta$};

\foreach \from/\to in {d/a,d2/b2,b/g}
\draw [->] (\from) -- (\to);

\foreach \from/\to in {d/b}
\draw [dotted, ->] (\from) -- (\to);

\path (b) edge [loop above, dotted, ->,out=45,in=135, looseness= 6]  (b);

\path (b2) edge [loop above, ->,out=45,in=135, looseness= 6]  (b2);

\def\x{30}
\path (a) edge [->,out=-\x,in=90+\x, ]  (g);

\path (g) edge [->,out=180-\x,in=-90+\x, ]  (a);

\path (a) edge [ dotted, ->,out=215,in=135, ]  (d);

\path (g) edge [ dotted, ->,out=45,in=-45, ]  (b);

\end{tikzpicture}
\caption{The first step of the TTC procedure (dotted edges denote second preferences)}
\label{fig:ttc}
\end{center}

\end{figure}
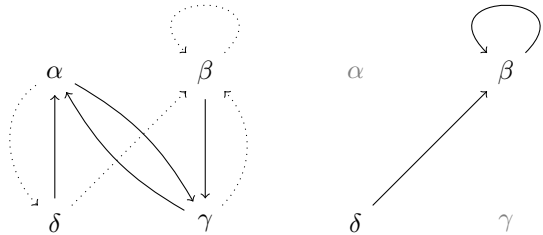

\paragraph*{Top Trading Cycles}

For housing markets, there is exactly one rule that is
simultaneously SP, PE and IR \cite{Sonmez1999}. 
The allocation that the rule
recommends can be obtained by following the {\em Top Trading
Cycles} procedure which we define below (see Figure~\ref{fig:ttc} for an example). In a housing market,
the endowments are singletons and so any IR rule must also
produce an assignment whose elements are singletons. We can assume
that each agent has a strict preference relation over $\mathcal{O}$.
We introduce the following useful notation: if $\alpha P_i \beta$ 
for all $\beta$ in some subset $\mathcal{O}'$ of the goods, we say
that agent $i$ {\em topranks} $\alpha$ in $\mathcal{O}'$ (if
$\mathcal{O}'=\mathcal{O}$ we say simply that $i$ topranks $\alpha$).\\

\noindent{\bf Top Trading Cycles}

\noindent{\bf Input:} An economy $\mathcal{E}=(N,\mathcal{O},\omega,\cR)$.

\noindent{\bf Output:} An assignment $z$.

\begin{enumerate}

	\item Create a directed graph $H_1$ whose vertex set is $V_1=\mathcal{O}$ with an edge 
	$(\alpha,\beta)$ in $E_t$ if and only if $a(\alpha)$ topranks $\beta$.
	
	\item For $t=1,2,\ldots$:
		\begin{enumerate}
	
		\item If $V_t$ is empty, stop.
		
		\item \label{step:arb} Otherwise, select an arbitrary cycle $(\gamma_1,\gamma_2,\ldots,\gamma_j)$ in $H_t$. 
			\begin{enumerate}
		
			\item Add $\gamma_1$ to $z_{a(\gamma_j)}$, and add $\gamma_{i+1}$ to $z_{a(\gamma_i)}$ 
			for $1 \leq i < j$.

			\item Let $V_{t+1}=V_t \setminus \{\gamma_1,\ldots,\gamma_j\}$.
			
			\item Let $H_{t+1}$ be the directed graph on $V_{t+1}$ with an edge $(\alpha,\beta)$ in $E_{t+1}$
			if and 
			only if $a(\alpha)$ topranks $\beta$ in $V_{t+1}$.
			
		\end{enumerate}
		
	\end{enumerate}

\end{enumerate}

In Step~\ref{step:arb}, an arbitrary cycle was selected.
Indeed, the order that cycles are removed from the graph
in TTC does not matter. However, it will be useful to refer
to the time at which goods are traded under TTC. In order
to make this notion well-defined, we insist that an economy
$\mathcal{E}$ is equipped with a total ordering over 
$\mathcal{O}$; we can refer to the {\em first} good in $\cO$. 
Observe that for each good $\alpha$ in $V_t$ there is a unique
directed walk with no repeated edges starting at $\alpha$; we call
this the {\em trading walk} starting at $\alpha$.
Since every element of $V_t$ has 
outdegree 1, this walk must contain a cycle. 
We define the cycle in Step~\ref{step:arb} to be
the one contained in the trading walk starting
at the first good in $V_t$. 
We can now define
the trading time $tt_{\cE}(\alpha)$ of a good $\alpha$ in a run of TTC on $\cE$ to be the 
least integer $t$ such that $\alpha \in V_t \setminus V_{t+1}$.
When the economy is unambiguous, we write $tt(\alpha)=tt_{\cE}(\alpha)$.
The following simple observations will be very useful later.

\begin{observation}\label{obs:tt}
	Suppose $\alpha$ and $\beta$ are goods in $V_t$ during a run of TTC. If the owner of $\alpha$ topranks $\beta$ in $V_t$,  then $tt(\alpha)\geq tt(\beta)$.
\end{observation}

\begin{observation}\label{obs:tt2}
Suppose $\alpha_i \in \omega_i$ and $\alpha_j \in \omega_j$. If $\alpha_i R_i \alpha_j$ and $\alpha_j R_j \alpha_i$. Then $tt(\alpha_i) \not= tt(\alpha_j)$.
\end{observation}

\begin{observation}\label{obs:tt3}
Suppose $\alpha$ is a good in $V_t$, and the trading
walk $W$ starting at $\alpha$ in $H_t$ is not a trading cycle.
Let $\beta$ be a good on the trading cycle in $W$. Then $tt(\beta) < tt(\alpha)$. 
\end{observation}

\begin{observation}\label{obs:tt4}
Suppose $\beta \in \text{TTC}_i(R)$ and $\alpha P_i \beta$. Then $tt(\alpha) < tt(\beta)$.
\end{observation}

Observe that TTC (as we have described it above) does not
require the endowments to be singletons. In other words,
it can be applied in the setting of generalised housing 
markets.
However, we know that in this case TTC is not strategy-proof
in general. For example, let $\cE=(N,\cO,\omega,R)$ be
the economy shown in Figure~\ref{fig:ttc}.
In this economy, $a(\delta)$ and $a(\gamma)$ have the same
preferences. Suppose that $a(\delta)$ and $a(\gamma)$ are the same;
let $\omega_1=\{\gamma,\delta\},\omega_2=\{\alpha\},\omega_3=\{\beta\}$
for instance. If agent 1 prefers the bundle $\{\alpha,\beta\}$ to 
its assignment $\{\alpha,\delta\}$ there is a possibility for agent 1
to benefit by misreporting their preferences. Agent 1 can report a
preference relation $R'_1$ such that $\beta P'_1 \alpha P'_1 \gamma P'_1 \delta$.
It is easy to verify that the allocation $\text{TTC}_1(R'_1,R_{-1})$ 
when agent 1 reports $R'_1$ is
$\{\alpha,\beta\}$.

In Figure~\ref{fig:lie} we see a more complicated example.
We adopt the convention throughout the paper that
\begin{tikzpicture} \draw[->] (0,0)--(1,0.1); \end{tikzpicture} denotes a first preference, \begin{tikzpicture} \draw[dotted,->] (0,0)--(1,0.1); \end{tikzpicture} denotes second preference, \begin{tikzpicture} \draw[dashed,->] (0,0)--(1,0.1); \end{tikzpicture} denotes third preference, \begin{tikzpicture} \draw[dash dot,->] (0,0)--(1,0.1); \end{tikzpicture} denotes fourth preference, and thereafter a dashed line with $i$ dots denotes the $(3+i)$th preference. 
We set $\omega_1=\{e_0,e_\alpha,e_\beta\}$.
The preferences of agent 1 are such that any bundle
including both $\alpha$ and $\beta$ is preferable to 
any bundle including one or the other or neither.
Informally, agent 1 wants to get $\alpha$ and $\beta$.
However, the order of preference of the individual goods,
according to $R_1$ is $\alpha,\beta,e_\alpha,e_\beta,e_0,\gamma,x,y$.
In the first round of TTC, the goods $\alpha,\gamma,e_0$ form a trading cycle. It can be seen that after these goods are removed,
first $x,y$ form a trading cycle, and then $\beta$ forms a trading cycle. Thus the assignment to agent 1 is $\{\alpha,e_\alpha,e_\beta\}$. Agent 1 has an incentive to lie; even though $x$ is not preferable to any individual good in $\omega_1$, obtaining it prevents
$x,y$ from forming a trading cycle. 

Since manipulation of TTC is clearly possible with multiple 
endowments, it is necessary to consider relaxing the
strict condition that a rule is strategy-proof.
Instead, we consider requiring that computing a beneficial
misreport is computationally intractable.

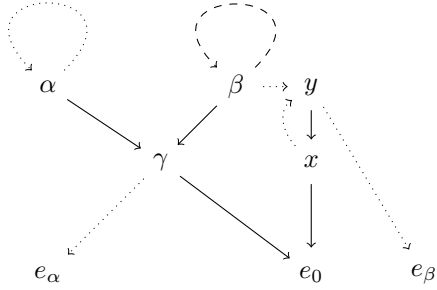
\begin{figure}[t]
\begin{center}
\begin{tikzpicture}
\tikzstyle{every node} = [circle, fill=white]

\node (eb) at (0,0) {$e_\alpha$};

\node (e0) at (3.5,0) {$e_0$};

\node (ea) at (5,0) {$e_\beta$};

\node (g) at (1.5,1.5) {$\gamma$};

\node (y) at (3.5,1.5)  {$x$};

\node (b) at (0,2.5) {$\alpha$};

\node (a) at (2.5,2.5) {$\beta$};

\node (x)  at (3.5,2.5) {$y$};

\foreach \from/\to in {b/g,a/g,x/y,g/e0,y/e0}
\draw [->] (\from) -- (\to);

\foreach \from/\to in {a/x,x/ea,g/eb}
\draw [dotted, ->] (\from) -- (\to);

\path (b) edge [loop above, dotted, ->,out=45,in=135, looseness= 10]  (b);

\path (a) edge [loop above, dashed, ->,out=45,in=135, looseness= 8]  (a);

\path (y) edge [ dotted, ->,out=135,in=215, ]  (x);

\end{tikzpicture}
\end{center}
\caption{Agent 1 (who owns $e_0,e_\alpha,e_\beta$) can get $\alpha$ and $\beta$ by trading $e_0$ for $x$.}
\label{fig:lie}
\end{figure}

\paragraph*{Computational complexity}
We are interested in the complexity of the following problem:\\

{\bf BENEFICIAL MISREPORT($\phi$)}

{\bf INPUT}: A generalised housing market economy $\cE$

{\bf QUESTION}: Does agent 1 have a beneficial misreport under $\phi$?\\

For simplicity's sake, we always assume that agent 1 is the would-be liar.
Since we are mainly interested in proving (conditional) lower bounds on
the complexity of manipulating TTC, we focus on the decision version of
the problem. Fujita {\em et al.} \cite{FujitaEA2015} showed that BM(TTC)
is $\NP$-complete in general (they refer to {\em Augmented} Top Trading Cycles,
but the description is equivalent). This result suggests that TTC might yet
be of practical use despite not being SP; an agent with limited computational
resources would have no incentive to lie. 

However, the hardness established by this result seems to depend
heavily on the size of the endowments. Indeed, the proof makes
the implicit assumption that one agent may have a number of goods
that grows with the number of agents.
This strongly suggests that
a parameterized approach is more appropriate. In fact, the $\NP$-completeness
of the problem could be very misleading; as we shall see later,
there is a polynomial time solution to the problem when the size
of the endowment is a fixed constant. 

\paragraph*{Parameterized complexity}

For a full treatment of the topic of parameterized
complexity, we refer the reader to the textbook by Downey
and Fellows \cite{DF1999}. We give a brief overview aimed at
non-specialists. Consider the following decision problems.\\
\begin{table}[h] \begin{tabular}{p{0.4\linewidth}
	p{0.4\linewidth}} {\bf CLIQUE}
	& {\bf VERTEX COVER} \\ {\bf INPUT}: A graph $G$ and
	an integer $K$
	& {\bf INPUT}: A graph $G$ and an integer $K$ \\
	{\bf QUESTION}: Is there a set $C$ of $K$ vertices
	of $G$ every pair of which is adjacent?	& {\bf
	QUESTION}: Is there a set $C$ of $K$ vertices of $G$
	such that every edge of $G$ contains a vertex of
	$C$? \\

\end{tabular}
\end{table}

Both of these problems are $\NP$-hard, which means that if
they can be solved in an amount of time that is polynomial
{\em in the total size of the input $(G,K)$} then $\Poly =
\NP$. On the other hand, when $K$ is fixed, and not part of
the input of the problem, both can be solved in polynomial
time. Indeed, CLIQUE can be solved in $|G|^{O(K)}$ time, and
VERTEX COVER can be solved in $O(2^K |G|)$ time. It should
be clear that there is a big difference in these run times:
$2^{20} \times 1000$ operations will take a modern computer
mere seconds whereas $1000^{20}$ is larger than the number
of atoms in the observable universe. This 2-dimensional
approach shows us that the complexity of these problems is
very sensitive to the size of the solution sought, and in
general a problem's complexity may depend heavily on the
size of some {\em parameter} in a way that classical
complexity ignores. 

We define a parameterized language to be a subset of
$\Sigma^* \times
\mathbb{N}$ for some alphabet $\Sigma$. If $(x,k)$ is a
member of a
parameterized language $L$ we say that $k$ is the {\em
parameter}. If there
exists an algorithm which can decide whether $(x,k)$ belongs
to $L$ in
$|x|^{f(k)}$ for some computable function $f$, then $L$ is
in the complexity
class $\XP$. If, additionally, there exists an algorithm
that decides
membership of $L$ in time $f(k)\cdot |x|$ for an arbitrary
function $f$, then
$L$ is in the complexity class $\FPT$. The above discussion
shows that VERTEX COVER
is in $\FPT$, but CLIQUE is thought not to be.

The fact that CLIQUE is $\NP$-hard is a {\em conditional lower 
bound} for the run time of an algorithm that solves CLIQUE.
Since every $\NP$ problem reduces to CLIQUE, a polynomial
time solution to this problem implies $\Poly=\NP$.
An analogous conditional lower bound exists in the parameterized
setting. We define the class of $\W[1]$ parameterized languages
to be those that reduce to the following.\\

{\bf SHORT NONDETERMINISTIC TURING MACHINE HALTING}

{\bf INPUT}: A nondeterministic Turing machine $M$

{\bf PARAMETER}: $k$

{\bf QUESTION}: Is it possible for $M$ to reach a halting state in at most $k$ steps?\\

It is considered extremely unlikely that $\FPT = \W[1]$.
We remark that CLIQUE happens to be $\W[1]$-hard (in 
fact $\W[1]$-complete) and thus unlikely to be in $\FPT$.
There is a whole hiearchy of classes $\FPT \subseteq \W[1] \subseteq \W[2] \subseteq \ldots \subseteq \XP$ (deciding the existence of a
dominating set of size $k$ is a $\W[2]$-complete problem for instance) and the inequality $\FPT \subset \XP$ is known to be strict. For our purposes, it is enough to consider $\W[1]$-hardness
as a conditional lower bound on the complexity of a decision problem.

Our main result is that BM(TTC), parameterized
by the size of the endowments, is $\W[1]$-hard.
For the rest of the paper, we refer only to 
the parameterized version of BM(TTC) 
In fact, BM(TTC) remains hard under a strong restriction
on the preference relations.

\paragraph*{Preference domains}

A preference relation $R_{i}$ is \textbf{lexicographic }if for each
$X,Y\subseteq\mathcal{O}$, $X\,P_{i}\,Y$ iff there is $b\in\mathcal{O}$
such that 1) $b\in X\backslash Y$, 2) for each $a\in\mathcal{O}$
with $a\,P_{i}\,b$, $a\in X\cap Y$. The lexicographic, additive,
responsive, and monotonic domains are ordered by inclusion \footnote{A preference relation $R_{i}$ is \textbf{monotonic }if
for each $X,Y\subseteq O$ such that $Y\subseteq X$, $X\,R_{i}\,Y$.
A preference relation $R_{i}$ is \textbf{responsive }if for each
$X\subset\mathcal{O}$, and each $a,b\in\mathcal{O}\backslash X$,
$X\cup\{a\}\,R_{i}\,X\cup\{b\}$ iff $a\,R_{i}\,b$. A preference
relation $R_{i}$ is \textbf{additive }if there is $u_{i}:\mathcal{O}\rightarrow\mathbb{R}$
such that for each $X,Y\subseteq\mathcal{O}$, $X\,R_{i}\,Y$ iff
$\sum_{a\in X}u_{i}(a)\geq\sum_{a\in Y}u_{i}(a)$.}. Our result holds even on the lexicographic domain; it immediately holds for
the more general domains.

For the rest of this paper, we assume that all agents
have lexicographic preference relations unless stated otherwise.
This allows us to write the preference relations in a compressed
format. We may write the preference relation of an agent $i$ as a
list of the singletons ordered by $R_i$, up to and including the least preferred element of $\omega_i$.

\section{The Main Result}


In order to prove that BM(TTC) is $\W[1]$-hard, we introduce
an auxiliary problem; our reduction is ultimately from the 
following problem.\\

{\bf MULTICOLOUR CLIQUE}

{\bf INPUT}: A graph $G$ with a proper vertex colouring $\phi$

{\bf PARAMETER}: The number of colours $k$

{\bf QUESTION}: Does there exists a clique of size $k$ 
in $G$?\\

In order to simplify our exposition, we introduce
an intermediate problem which we will prove is $\W[1]$-hard
and reduce to BM(TTC).
In a directed graph $G$ with a proper vertex colouring
$\phi:V(G) \rightarrow [k]$, an edge $(u,v)$ with
$\phi(v)=\phi(u)+1$ will be called a {\em rung}. On the
other hand if $\phi(v)<\phi(u)$, we say that $(u,v)$ is a
{\em snake}. A {\em partial ladder} in such a graph is a set of $k$
vertices $\{v_1,\ldots,v_j\}$ such that $(v_i,v_{i+1})$ is a
rung for all $1\leq i \leq j-1$. A (partial) ladder is {\em snakeless}
if there are no snakes between its vertices. We define the
problem of deciding the existence of a snakeless ladder in a
$k$-coloured graph as follows:\\

\begin{figure}[t]
\begin{center}
\begin{tikzpicture}[node distance=4cm]

\tikzstyle{every node} = [draw];

\node[circle,fill=black,text=white] (v1) at (0,2) {$v_1$};

\node[circle,fill=orange] (v2) at (2,2) {$v_2$};

\node[circle,fill=cyan] (v3) at (2,0) {$v_3$};

\node[circle,fill=yellow] (v4) at (0,0) {$v_4$};

\foreach \from/\to in {v1/v2,v2/v3,v3/v4}
\draw [->] (\from) -- (\to);

\foreach \from/\to in {v3/v1}
\draw [->,color=red] (\from) -- (\to);
\end{tikzpicture}
\caption{A directed $4$-coloured graph $G$ (snake in red).}
\label{fig:snl}
\end{center}
\end{figure}
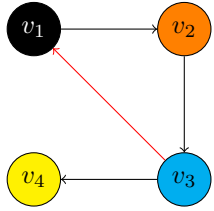

{\bf SNAKELESS LADDER}

{\bf INPUT}: A directed graph $G$ with a proper vertex colouring $\phi$

{\bf PARAMETER}: The number of colours $k$

{\bf QUESTION}: Does $G$ contain a snakeless ladder?\\

\begin{lemma}
SNAKELESS LADDER is $\W[1]$-hard.
\end{lemma}

\begin{proof} The proof is a simple reduction from
	MULTICOLOUR CLIQUE. From a $k$-coloured graph $G$ we
	obtain a directed graph $G'$ by first complementing
	the set of edges between non-adjacent colour
	classes. In other words if $|\phi(u)-\phi(v)| \not =
	1$, then $uv$ is an edge in $G'$ if and only if it
	was not an edge in $G$. We then direct the edge $uv$
	from $u$ to $v$ if $\phi(v)=\phi(u)+1$ or if
	$\phi(v)<\phi(u)$. Now consider a set of vertices
	$\{v_1,\ldots,v_k\}$ in $G'$. By definition
	$(v_1,v_2),(v_2,v_3)\ldots,(v_{k+1},v_k)$ are rungs
	in $G'$ if and only if $v_1 v_2,v_2 v_3,\ldots,
	v_{k+1} v_k$ are edges in $G$. Similarly, for $i<j$
	there is no snake $(v_j,v_i) $ in $G'$ if and only
	if $v_i v_j$ is an edge of $G$. Therefore,
	$\{v_1,\ldots,v_k\}$ form a snakeless ladder in $G'$
	if and only if they form a clique in $G$.

\end{proof}

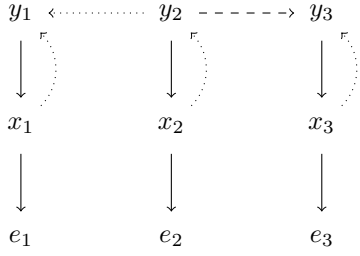
\begin{figure}[t]
\begin{center}
\begin{tikzpicture}
\tikzstyle{every node} = [circle, fill=white]

\node (e1) at (3,0) {$e_1$};

\node (e2) at (5,0) {$e_2$};

\node (e3) at (7,0) {$e_3$};

\node (x11) at (3,1.5) {$x_1$};

\node (x22) at (5,1.5) {$x_2$};

\node (x33) at (7,1.5) {$x_3$};

\node (y11) at (3,3) {$y_1$};

\node (y22) at (5,3) {$y_2$};

\node (y33) at (7,3) {$y_3$};

\foreach \from/\to in {y11/x11,y22/x22,y33/x33,x11/e1,x22/e2,x33/e3}
\draw [->] (\from) -- (\to);

\foreach \from/\to in {y22/y11}
\draw [dotted, ->] (\from) -- (\to);

\draw [dashed, ->] (y22) -- (y33);

\foreach \from/\to in {x11/y11,x22/y22,x33/y33}
\draw [dotted, ->] (\from) to[out=45, in=-45] (\to);

\end{tikzpicture}
\caption{Three vertex gadgets.}
\label{fig:gadget}
\end{center}
\end{figure}

\begin{theorem}
BM(TTC) is $\W[1]$-hard, even when the preference domain is lexicographic
\end{theorem}

\begin{proof}

	We reduce from SNAKELESS LADDER. Let $G$ be a
	graph with a proper $k$-colouring for some integer
	$k$. We will obtain an economy $\cE_G$ and an
	integer $k'$ such that the endowment of player 1 in
	$\cE_G$ has size $k'$ and furthermore player 1 has a
	beneficial misreport $R_1'$ if and only if there is
	a snakeless ladder in $G$. We will then
	show that $\cE_G,k'$ can be constructed in time
	$f(k)|G|^c$ for some function $f$ and constant $c$.

We begin the construction of $\cE_G=(N,\cO,\omega,R)$ by
setting $\omega_1={e_\alpha,e_\beta}\cup \{e_j:1 \leq
	j\leq k\}$. We can assume that all other agents
	have a singleton endowment. For each vertex $v_i$ in $G$,
	we add to $\cE_G$ a {\em vertex gadget}. A vertex
	gadget is a pair of goods $x_i,y_i$ and their respective owners,
	who have particular preferences depending on the 
	colour and neighbourhood of $v_i$. The agent $a(y_i)$ always
	topranks $x_i$, and $a(x_i)$ topranks $e_j$ where $j$ is the
	colour of $v_i$. The full preference relation of $a(x_i)$ 
	is $(e_j,y_i,x_i)$. The preferences of $a(y_i)$ are as follows.
	Let $v_{s_1},v_{s_2},\ldots$ be the
	endpoints of the snakes which start at $v_i$ and let $v_{r_1},v_{r_2},\ldots$ be the endpoints of the rungs
	which start at $v_i$.
	For a vertex $v_i$ of colour $k$, the full preference relation
	of $a(y_i)$ is $(x_i,y_{s_1},y_{s_2},\ldots,e_\beta,y_i)$.
	For a vertex $v_i$ not of colour $k$, the full preference relation of $a(y_i)$ is $(x_i,y_{s_1},y_{s_2},\ldots,y_{r_1},y_{r_2},\ldots,y_i)$.
	Note that $a(y_i)$ prefers each of the goods representing
	the snakes of $v_i$ to each of the goods representing the rungs.
	In Figure~\ref{fig:gadget} we see three vertex gadgets.
	For $i=1,2,3$ we have that $v_i$ has colour $i$ and we see
	that $(v_2,v_1)$ is a snake and $(v_2,v_3)$ is a rung.
	
	We continue our construction by the addition of three goods
	(and their respective owners) $\alpha,\beta,\gamma$.
	The full preference relation of $a(\alpha)$ is $(\gamma,\alpha)$.
	The full preference relation of $a(\gamma)$ is $(e_1,e_2,\ldots,e_k,e_\alpha,\gamma)$. 
	The preference relation of $a(\beta)$
	is as follows.
	Suppose $v_1,v_2,\ldots,v_j$ are the vertices of colour 1
	in $G$ in an arbitrary order. The full preference relation
	of $a(\beta)$ is $(\gamma,y_1,y_2,\ldots,y_j,\beta)$. 
	Note that $a(\beta)$ prefers all goods representing
	vertices of colour 1 to $\beta$ and that $\beta$ is preferred to 
	all other goods representing vertices. 
	
	The construction is completed by revealing the true
	preference relation of agent 1; namely, $(\alpha,\beta,e_\alpha,e_\beta,e_k,e_{k-1},\ldots,e_1)$. 
	For clarity, we have provided Table~\ref{tab:pref} which
	shows the preferences of the agents described above. 
	In Figure~\ref{fig:econ} we see an example of an economy
	constructed from the graph $G$ in Figure~\ref{fig:snl}.
	The preferences of agent 1 are omitted.
	When agent 1 reports the truth, the assignment
	received is $\{\alpha,e_\alpha,e_\beta,e_k,\ldots,e_2\}$.
	
	\begin{claim}\label{cla:both}
		If there is a beneficial misreport $R'_1$ for agent 1,
	then $\alpha,\beta\in\text{TTC}_1(R'_1,R_{-1})$
	\end{claim}
	
	\begin{proof}
		The two best goods obtained by agent 1 by reporting the truth
		are $\alpha$ and $e_\alpha$. By the lexicographic 
		property of $R_1$, any bundle preferred by agent 1 to its true
		assignment must include $\alpha$ since agent 1 topranks $\alpha$.
		Similarly, a preferred bundle must include a good that is preferred by agent 1 to $e_\alpha$.
		The only such good is $\beta$.
	\end{proof}
		
	The only bundles agent 1 prefers to their true assignment
	include both $\alpha$ and $\beta$.
	The reader may wish to verify that there is no beneficial
	misreport available to agent 1 in the economy in Figure~\ref{fig:snl}.
	This is in contrast with Figure~\ref{fig:lie}, where
	agent 1 was able to prevent $x,y$ forming a trading
	cycle by obtaining $x$, and therefore obtain both $\alpha$
	and $\beta$.
	If agent 1 tries the misreport $(x_1,x_2,x_3,x_4,\alpha,\beta)$
	it is easy to see that $y_1,y_2,y_3$ will at some point
	form a trading cycle. Thus $\beta$ will form a trading cycle
	on its own and not be included in the assignment to
	agent 1. 
	This is because $v_1,v_2,v_3,v_4$ is not a snakeless ladder
	in $G$. 
	If the snake $(v_3,v_1)$ was omitted from $G$,
	then $G$ would have a snakeless ladder 
	and $y_3$ would no longer prefer $y_1$ to $y_4$.

\begin{table}[t]
\begin{center}
\begin{tabular}{l  l  l  l  l  l }
1	  & $\alpha$ & $\beta$  & $\gamma$   & $x_{i}$ & $y_{i}$  \\

\hline

$\alpha$  & $\gamma$ & $\gamma$	& $e_1$      & $e_j$	& $x_{i}$  \\

$\beta$	  & $\alpha$ & \{$y_i$ of colour 1\} & $e_2$      & $y_{i}$ & \{snakes\} \\

$e_\alpha$	  &  	     & $\beta$	& $\cdots$   & $x_{i}$	& \{rungs\} \\

$e_\beta$  &          &    	& $e_k$	     & &	($e_\beta$ if $j=k$) \\

$e_k$	  &          &		& $e_\alpha$ &			&	$y_{i}$		 \\

$\cdots$&	     &		& $\gamma$   &		&		\\

$e_1$ &	     &		&	     &	&			

\end{tabular}
\caption{The preferences of some agents in $\cE_G$ ($v_i$ has colour $j$)}
\label{tab:pref}
\end{center}
\end{table}

	We now formalise this intuition and argue that
	agent 1 has a beneficial misreport in $\cE_G$ if
	and only if $G$ has a snakeless ladder.
	Suppose that $L=(v_1,v_2,\ldots,v_k)$ is a
	snakeless ladder in $G$ (so $v_i$ has colour $i$ in $G$).
	We claim that $R'_1=(x_1,x_2,\ldots,x_k,\alpha,\beta,e_1,\ldots,e_k,e_\alpha,e_\beta)$ is a beneficial misreport.
	We abuse our terminology slightly and say that if $v_i$ is
	of colour $j$, then $x_i,y_i$ are {\em goods of colour $j$}.
	Let $\cE'_G$ be the economy obtained from $\cE_G$
	by replacing $R_1$ by $R'_1$, and consider a 
	run of TTC on $\cE'_G$. For the rest of this proof, we will write $tt(\delta)=tt_{\cE'}(\delta)$ for the trade time of a good
$\delta$ during a run of TTC on $\cE'$. It is clear that
	in $H_1$, there is a trading cycle $e_1,x_1$.
	After this is removed, there will be a trading
	cycle $e_2,x_2$. Observe that of all the
	vertices of colour 1 and 2 in $G$,
	only $v_1$ and $v_2$ are represented 
	by goods in $H_{tt(e_2)+1}$. Furthermore,
	$a(y_1)$ topranks $y_2$ in $V_{tt(e_2)+1}$.
	Now $e_3,x_3$ form a trading cycle, and
	after this is removed, only $y_3$ remains 
	among the goods of colour 3.
	Since $L$ is a snakeless ladder, $a(y_3)$
	prefers each good of colour 4
	to $y_1$. 
	Similarly, we have that $a(y_i)$ topranks $a(y_{i+1})$
	in $H_{tt(e_k)+1}$ for $1\leq i < k$. 
	Since $v_k$ is of colour $k$, $y_k$ topranks $e_\beta$.
	The goods $x_i, 1\leq i \leq k$ are not 
	in $V_{tt(e_k)+1}$, and so agent 1 topranks $\alpha$
	in $H_{tt(e_k)+1}$ (according to the false
	preference relation $R'_1$), and $a(\gamma)$ topranks
	$e_\alpha$.
	Thus, $e_\alpha,\alpha,\gamma$ form a trading cycle.
	The assignment of agent 1 under $R'_1$ includes $\alpha$ as required.
	Since $\gamma$ is not in $V_{tt(\alpha)+1}$, so $a(\beta)$
	topranks $y_1$, the only remaining good of colour $1$. Furthermore, $\alpha$ is
	not in $V_{tt(\alpha)+1}$, so agent 1 topranks $\beta$
	in $H_{tt(\alpha)+1}$.
	Finally, $e_\beta,\beta,y_1,\ldots,y_k$ form a trading
	cycle, and $\beta$ is included in the assignment to agent 1.
	
	Suppose instead that there exists a beneficial misreport $R'_1$ for agent 1.
	We show that the trading cycle in $H_{tt(\beta)}$
	that contains $\beta$ is of the form $(\beta,y_1,y_2,\ldots,y_k,e_\beta)$ where $v_1,v_2,\ldots,v_k$ is a snakeless
	ladder in $G$.
	We now demonstrate that $\alpha P'_1 \beta$ must hold.
	By Observation~\ref{obs:tt}, we have that $tt(\gamma) \leq tt(\alpha)$.  On the other hand, if $tt(\gamma) < tt(\alpha)$, 
then $a(\alpha)$ topranks $\alpha$ in $V_{tt(\gamma)+1}$; thus $a(\alpha)$ keeps $\alpha$, a contradiction. This shows that $tt(\alpha)=tt(\gamma)$. Observation~\ref{obs:tt} also gives
us that $tt(\gamma) \leq tt(\beta)$, and by Observation~\ref{obs:tt2} this inequality is strict. This shows that $tt(\alpha)<tt(\beta)$.
Since $\alpha$ and $\beta$ are both in the assignment to agent 1 by
assumption, we must have $\alpha P'_1 \beta$. 

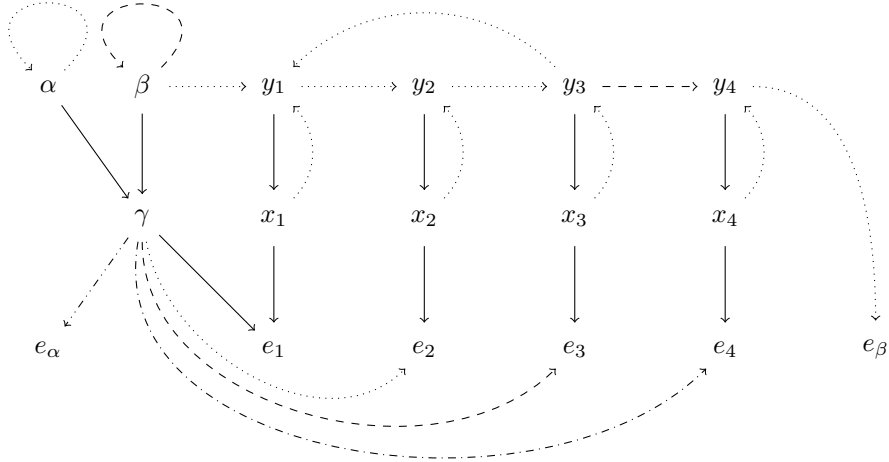
\begin{figure}[t]
\begin{center}
\begin{tikzpicture}
\tikzstyle{every node} = [circle, fill=white]
\def\h{1.75}
\node (ea) at (0,0) {$e_\alpha$};

\node (a) at (0,2*\h) {$\alpha$};

\node (b) at (1.25,2*\h) {$\beta$};

\node (g) at (1.25,\h) {$\gamma$};

\node (e1) at (3,0) {$e_1$};

\node (e2) at (5,0) {$e_2$};

\node (e3) at (7,0) {$e_3$};

\node (e4) at (9,0) {$e_4$};

\node (eb) at (11,0) {$e_\beta$};

\node (x11) at (3,\h) {$x_{1}$};

\node (x22) at (5,\h) {$x_{2}$};

\node (x33) at (7,\h) {$x_{3}$};

\node (x44) at (9,\h) {$x_{4}$};

\node (y11) at (3,2*\h) {$y_{1}$};

\node (y22) at (5,2*\h) {$y_{2}$};

\node (y33) at (7,2*\h) {$y_{3}$};

\node (y44) at (9,2*\h) {$y_{4}$};

\foreach \from/\to in {b/g,a/g,g/e1,y11/x11,y22/x22,y33/x33,y44/x44,x11/e1,x22/e2,x33/e3,x44/e4}
\draw [->] (\from) -- (\to);

\foreach \from/\to in {b/y11,y11/y22,y22/y33}
\draw [dotted, ->] (\from) -- (\to);

\foreach \from/\to in {x11/y11,x22/y22,x33/y33,x44/y44}
\draw [dotted, ->] (\from) to[out=45, in=-45] (\to);

\draw [dotted, ->] (y33) to[out=135, in=45] (y11);

\draw [dashed, ->] (y33) -- (y44);


\draw [dotted, ->] (y44) to[out=0, in=90] (eb);

\draw [dotted, ->] (g) to[out=-80, in=-135] (e2);

\draw [dashed, ->] (g) to[out=-90, in=-135] (e3);

\draw [dash dot, ->] (g) to[out=-100, in=-135] (e4);

\draw [dash dot dot, ->] (g) -- (ea);

\path (a) edge [loop above, dotted, ->,out=45,in=135, looseness= 10]  (a);

\path (b) edge [loop above, dashed, ->,out=45,in=135, looseness= 8]  (b);

\end{tikzpicture}
\caption{The economy $\cE_G$ associated with the graph in Figure~\ref{fig:snl}}
\label{fig:econ}
\end{center}
\end{figure}

Observation~\ref{obs:tt} also tells us that $tt(e_1) \leq tt(\gamma)$. We show that this inequality is strict
If $tt(e_1)=tt(\gamma)$, then $e_1,\alpha,\gamma$ form a 
trading cycle in $H_{tt(e_1)}$. Then in $H_{tt(e_1)+1}$, each
pair of colour $1$ forms a trading cycle. 
By Observation~\ref{obs:tt3}, no pair of colour $1$ is in
$H_{tt(\beta)}$, and $\beta$ forms a trading cycle with itself,
contradicting the definition of $R'_1$.

In $H_{tt(e_1)+1}$, the agent $a(\gamma)$ topranks $e_2$.
Again, by Observation~\ref{obs:tt} we have that $tt(e_2) \leq tt(\gamma)$. A very similar argument to the above shows
that this inequality is strict. Indeed,
since $\beta$ does not form a trading cycle on its own by
assumption, there must be at least one good $y_i$ of colour 1
in $H_{tt(\beta)}$. Since there must be a trading cycle
including $\beta$ in $H_{tt(\beta)}$, the good $x_i$ cannot
be in $H_{tt(\beta)}$. If $tt(e_2)=tt(\gamma)$, every
pair of colour 2 forms a trading cycle in $H_{tt(e_2)+1}$.
Thus no pair of colour 2 is in $H_{tt(\beta)}$, so $y_i$
forms a trading cycle with itself a contradiction.

Proceeding by induction, we see that $tt(e_i)<tt(\alpha)$ for
$1\leq i \leq k$. Consider $H_{tt(\alpha)}$. For $1 \leq i \leq k$,
the goods $e_i$ are not in $V_{tt(\alpha)}$. The only other good
that $a(\gamma)$ ranks above $\gamma$ is $e_\alpha$. 
Thus the trading cycle in $H_{tt(\alpha)}$ containing 
$\alpha$ must be $\alpha,\gamma,e_\alpha$.

Now consider $H_{tt(\beta)}$. Suppose $y_i$ is a good of
colour $j$ in $V_{tt(\beta)}$. Then, without loss of generality,
$x_i$ is not in $V_{tt(\beta)}$, since we have shown $e_j$ is not, and the
order in which cycles are removed is arbitrary.
Moreover, suppose there are distinct goods $y_i,y_{i'}$ of colour $j$ in $V_{tt(\beta)}$.
By Observation~\ref{obs:tt}, $tt(e_j) \leq tt(x_i)$ and $tt(e_j) \leq tt(x_{i'})$. By Observation~\ref{obs:tt2}, $tt(x_i) \not= tt(x_{i'})$. Thus $e_j$ must have been in a trading cycle with at most one of $x_i,x_{i'}$. Without loss of generality, $tt(e_j)$ is strictly less than $tt(x_i)$, and $y_i,x_i$ form a trading cycle in $H_{tt(e_j)+1}$, a contradiction.

The trading walk
in $H_{tt(\beta)}$ starting at $\beta$ must be a trading cycle.
Since $\gamma$ is not in 
$V_{tt(\beta)}$, there must be exactly one good $y_i$ of colour $1$
in $V_{tt(\beta)}$. Without loss of generality, that good is 
$y_1$. 
The agent $a(y_1)$ only ranks goods of the form $y_j$ of colour
2 above $y_1$ in $V_{tt(\beta)}$. As we have discussed, there must be
exactly one such good; without loss of generality,
that good is $y_2$. We proceed by induction.
Suppose there is a path $P$ in $H_{tt(\beta)}$ of the form $(\beta,y_1,\ldots,y_i)$ where $y_i$ is of colour $i$. 
Observe that $v_1,\ldots,v_{i-1}$ must be a snakeless partial ladder,
though there may yet be a snake $(v_i,v_{i'})$ with $i'<i$.
However, if $a(y_i)$ topranks some
good $y_{i'} \in P$ with $i'<i$ then the trading walk starting at $\beta$ is not a cycle, which is a contradiction. All other goods
of colour less than $i$ are omitted from $V_{tt(\beta)}$, as is $x_i$. Thus $a(y_i)$ topranks a good of colour $i+1$ in $V_{tt(\beta)}$. Without loss of generality, that good is $y_{i+1}$. Observe that $v_1,ldots,v_i$ is a snakeless partial ladder, and that there is a path in $H_{tt(\beta)}$ of the form $(\beta,y_1,\ldots,y_{i+1})$.
We conclude that there is a path of the form $(\beta,y_1,\ldots,y_k)$ 
in $H_{tt(\beta)}$, and by the same argument, $y_k$ must toprank $e_\beta$. In other words, we have that $v_1,\ldots,v_k$ is a snakeless ladder as required.

\end{proof}

\section{An Upper Bound}

We leave the possibility of a matching upper bound on the complexity
of BM(TTC) (i.e. a proof of membership in $\W[1]$) as an interesting
open problem. We conclude the paper with an upper bound that
nevertheless represents a negative result for TTC. Informally,
if the size of the endowments is a fixed constant, BM(TTC) can be decided in polynomial time. In fact, our result is slightly stronger,
in that we present an explicit constructive algorithm that
can produce a beneficial misreport. We also highlight that this
result holds regardless of the preference domain.

\begin{proposition}
BM(TTC) is in $\XP$.
\end{proposition}

\begin{proof}
We will show that the following algorithm computes a beneficial misreport, if one exists, for agent 1 in time at most $k!n^{k+c}$ where $n$ is the number of goods, $k$ is the size of the endowment and $c$ is a constant associated with the runtime of TTC.\\

\noindent{\bf Algorithm $\cA$}

\noindent{\bf Input:} An economy $\mathcal{E}=(N,\mathcal{O},\omega,R)$.

\noindent{\bf Output:} An beneficial misreport for agent 1.

\begin{enumerate}

	\item Let $\omega_1=\{e_1,\ldots,e_k\}$ be the endowment of agent 1.
	
	\item Let $X_1,X_2,\ldots,X_m$ be the bundles of size $k$ such that
	 $X_i R_1 \omega_1$ for each $i$ (ordered according to $R_1$).
	 
	 \item For $i=1,2,\ldots,m$:
	 
	 	\begin{enumerate}
	 		
	 		\item Let $X_i=\{\gamma_1,\ldots,\gamma_k\}$
	 		
	 		\item Let $\alpha_1,\ldots,\alpha_{n-k}$ be the goods not in $X_i$.
	 		
	 		\item For each permutation $\{i_1,\ldots,i_k\}$ of $\{1,\ldots,k\}$:
	 		
	 			\begin{enumerate}
	 			
	 				\item Let $R'_1$ induce the following ordering over the singletons: $(\gamma_{i_1},\ldots,\gamma_{i_k},$ $\alpha_1,\ldots,\alpha_{n-k})$ (the order of $\alpha_1,\ldots,\alpha_{n-k}$ is arbitrary).

	 				\item\label{step:try} Let $z$ be the output of TTC($R'_1,R_{-1}$)
	 				
	 				\item If $z_1 \in z$ is the bundle $X_i$, then return $R'_1$
	 			
	 			\end{enumerate}
	 			
	 	\end{enumerate}
	 	
	 \item Return 0

\end{enumerate}

The correctness of our algorithm is a corollary of the following claim.

\begin{claim}
Suppose the assignment to agent 1 under TTC is $z_1=\{\gamma_1,\ldots,\gamma_k\}$ if they report $R'_1$, with $\gamma_i R'_1 \gamma_{i+1}$ for $1 \leq i < k$. Let $R''_1$ be a misreport
such that $\gamma_i R''_1 \gamma_{i+1}$ for $1 \leq i < k$ and 
$\gamma_k R''_1 \alpha$ for every good $\alpha \not \in z_1$.
Then the assignment to agent 1 if they report $R''_1$ is also $z_1$
\end{claim}

\begin{proof}
Let $\cE'$ and $\cE''$ be the (otherwise identical) economies in which agent 1 
reports $R'_1$ and $R''_1$ respectively. 
For $t=1,2,\ldots$ let $H'_t$ and $H''_t$ be the graphs
generated by running TTC on $\cE'$ and $\cE''$ respectively.
Let $tt'(\alpha)$ and $tt''(\alpha)$ be the trade times of
$\alpha$ during a run of TTC on $\cE'$ and $\cE''$ respectively.
By Observation~\ref{obs:tt4}, $tt'(\gamma_1)<tt'(\gamma_i)$ and $tt''(\gamma_1) < tt''(\gamma_i)$ for $1< i \leq k$.
Consider $H'_{tt'(\gamma_1)}$. Clearly, none of the cycles that
have been removed from $H'_1,H'_2,\ldots,H'_{tt'(\gamma_1)-1}$ have
included any of $\gamma_1,\ldots,\gamma_k$. Therefore we may assume
without loss of generality that the same cycles are removed
from $H''_1,\ldots,H''_{tt''(\gamma_1)-1}$ and thus $H''_{tt''(\gamma_1)}$ and $H'_{tt'(\gamma_1)}$ are identical.
Since $\gamma_1$ is assigned to agent 1 in $\cE'$, it must
also be assigned to agent 1 in $\cE''$. The claim follows by induction.

\end{proof}

Thus if there is a beneficial misreport such that the assignment
to agent 1 is $X=\{\gamma_1,\ldots,\gamma_k\}$, then it is enough to check only those misreports that rank the goods in $X$ above any other goods.

It remains for us to analyse the runtime of Algorithm~$\cA$. There are
at most $\binom{n}{k}\leq n^k$ bundles that agent 1 can prefer above the endowment. There are $k!$ different permutations of 
a bundle of size $k$. So Step~\ref{step:try} is performed at most
$k!n^k$ times. In this step, TTC is called. Since TTC takes polynomial time to perform (it takes at most $n$ steps to find a cycle, and at least one good is removed for each time step), there
exists a constant $c$ such that this step takes at most $n^c$ time. The overall run time is therefore at most $k!n^{k+c}$ as required.

\end{proof}

\section*{Acknowledgements}

A large part of this research took place while the first author was at the University of Helsinki, Finland. The project began with a chance meeting in the sauna of Töölö Towers; we thank the wonderful staff there. We also thank Jukka Suomela for many useful discussions. 

\bibliographystyle{plain}
\bibliography{bm}
\end{document}